\documentclass[12pt]{article}
\usepackage[mathscr]{eucal}
\usepackage{amssymb,latexsym}
\usepackage{verbatim}
\usepackage{amsmath}
\usepackage{amsthm}
\usepackage{hyperref}
\usepackage{enumerate}
\usepackage{graphicx}
\usepackage{tipa}

\usepackage{framed}
\usepackage{authblk}

\newtheorem{thm}{Theorem}[section]
\newtheorem{lem}[thm]{Lemma}
\newtheorem{Def}{Definition}[section]
\newtheorem{cor}[thm]{Corollary}

\newcommand\bbP{{\mathbb P}}
\newcommand\bbR{{\mathbb R}}
\newcommand\bbZ{{\mathbb{Z}}}

\DeclareMathOperator{\PD}{PD}

\renewcommand\S{\Sigma}

\renewcommand\d{\partial}

\newcommand\D{\nabla}

\newcommand\g{\gamma}

\renewcommand\th{\theta}

\newcommand{\bdry}{\partial}

\newcommand{\rk}{{\rm rk}}
\newcommand{\incl}{{\it incl}}

\newcommand\vs{\vspace}

\newcommand\beq{\begin{eqnarray}}
\newcommand\eeq{\end{eqnarray}}
\newcommand\ben{\begin{enumerate}}
\newcommand\een{\end{enumerate}}
\newcommand\bit{\begin{itemize}}
\newcommand\eit{\end{itemize}}

\DeclareFontFamily{OT1}{rsfs}{} \DeclareFontShape{OT1}{rsfs}{m}{n}{ <-7> rsfs5 <7-10> rsfs7 <10-> rsfs10}{}
\DeclareMathAlphabet{\mycal}{OT1}{rsfs}{m}{n}

\newcounter{mnotecount}[section]

\title{On the topology of initial data sets \\with higher genus ends}
\author{Kenneth L. Baker}
\author{Gregory J. Galloway}
\affil{Department of Mathematics \\University of Miami, Coral Gables, FL 33124}

\begin{document}

\date{}
\maketitle

\begin{abstract} 
In this note we study the topology of $3$--dimensional initial data sets with horizons  of a sort associated with asymptotically locally anti-de Sitter spacetimes.  We show that, within this class, those initial data sets which contain no (immersed) marginally outer trapped surfaces in their interior must have simple topology: they are a product of a surface and an interval, or a mild variation thereof, depending on the connectedness of the horizon and on its genus relative to that of the end.  The results obtained here extend results in \cite{EGP} to the case of higher genus ends. 
\end{abstract}

\section{Introduction}

One of the interesting features of general relativity is that it does not a priori impose any restrictions on the topology of space.
In fact, as was shown in \cite{IMP}, given an asymptotically flat initial data set of arbitrary topology, there always exists a solution to the vacuum Einstein constraint equations.  However, according to the principle of topological censorship, the topology of the domain of outer communications (DOC), i.e.\ the region outside of all black holes (and white holes), should, in a certain sense, be simple.  Roughly speaking, results on topological censorship \cite{FSW, Gdoc, GSWW, CGS} show that, under suitable energy and causality conditions, the topology of the DOC (at the fundamental group level) cannot be more complicated than the topology at infinity.  In particular, in the asymptotically Minkowskian case, the DOC must be simply connected.   However, the results alluded to here are spacetime results, i.e., they involve conditions that are essentially global in time.  In \cite{EGP} a result on topological censorship was obtained at the pure initial data level for asymptotically flat initial data sets, thereby circumventing  difficult questions of global evolution; cf. \cite[Theorem 5.1]{EGP}.

The aim of this note is to extend this result 
to $3$-dimensional initial data sets which arise in asymptotically locally anti-de Sitter (ALADS) spacetimes. Examples of such are the generalized Kottler  spacetimes \cite{Brill, Mann, CS}, which are solutions to the vacuum Einstein equations with negative cosmological constant.  These solutions have spacetime ends and event horizons with arbitrary (but equal) genus.  Cauchy surfaces for the DOC have product topology.  Here we will establish conditions  on ALADS initial data sets, similar to those in \cite{EGP}, which imply product topology.   The proofs rely on existence results for marginally outer trapped surfaces, as well as our current understanding of 
$3$--manifolds.   Results in \cite{CLR} will play a key role in establishing product topology. 

In Section \ref{prelim} we present some preliminary material.  Our main results are presented in Section \ref{results}.  

\section{Preliminaries}
\label{prelim}

We begin with some basic definitions and important facts about $3$--manifolds.  Let $V$ be a compact $3$--manifold without spherical boundary components.\footnote{Here and elsewhere we assume that boundary components (which physically correspond to horizons or ends) are of genus 
$\ge 1$.  The case of spherical ends is considered in \cite{EGP}.}
$V$ is said to be {\it irreducible} provided every embedded $2$--sphere in $V$ bounds a ball  in $V$.

A surface $\S$ embedded in $V$ is said to be {\em compressible} if there exists an embedded disk $D \subset V$ such that $D \cap \S = \bdry D$ and $\bdry D$ does not bound a disk in $\S$; such a disk is a {\em compressing disk} for $\S$.    A surface that is not compressible is {\em incompressible}.
It is a fundamental consequence of the Loop Theorem \cite{Papakyriakopoulos} that $\S$ is incompressible if and only if the map on fundamental groups $i_* \colon \pi_1(\S) \to \pi_1(V)$ induced by inclusion $i \colon \S \hookrightarrow V$ is injective. 

If $\S$ is a compressible surface in $V$ and $D$ is a compressing disk for $\S$, then $D$ has a collar neighborhood $N(D) \cong D^2 \times (0,1)$ such that $\S \cap N(D)  \cong \bdry D \times (0,1)$, and  the {\em compression} of $\S$ along $D$ is the surface $\S' = (\S \cup \bdry N(D)) - (\S \cap N(D))$.  If moreover $\S$ is a component of $\bdry V$, then this compression produces the submanifold $V' = V - N(D)$.  Observe that $V$ may be recovered from $V'$ by attaching a $3$--dimensional $1$--handle along $\S'$; the compressing disk is its {\it co-core}.

From Theorem 10.5 in \cite{Hempel}, together with the positive resolution of the Poincar\'e Conjecture (which ensures that there are no fake $3$-cells),
we have the following algebraic criterion for $V$ to be an $I$--bundle (where $I$ is the interval $[0,1]$).

\begin{thm}\label{thm:product} 
Let $V$ be a compact, connected, orientable $3$--manifold without spherical boundary components, and $\S$ an incompressible boundary component.  If the index, $[\pi_1(V) : i_* \pi_1(\S)]$, of $i_* \pi_1(\S)$ in $\pi_1(V)$ is finite, then either 
\begin{itemize}
\item  $[\pi_1(V) : i_* \pi_1(\S)]=1$ and $V$ is diffeomorphic to $
\S \times [0,1]$ with $\S = \S \times \{0\}$, or
\item  $[\pi_1(V) : i_* \pi_1(\S)]=2$ and $V$ is a twisted $I$--bundle over a compact non-orientable surface $\bar{\S}$ with $\S$ the associated $0$--sphere bundle.
\end{itemize}
\end{thm}

A group $G$ is
said to be {\it residually finite} if for each non-identity element $g \in G$, there is a normal subgroup
$N$ of finite index such that $g \notin N$.  It follows from work of Hempel \cite{Hempel2}, together with the positive resolution of the geometrization conjecture, that the fundamental group of every compact $3$--manifold $V$ is residually finite; see e.g.\ \cite{afw-survey}. Hence, by basic covering space theory, if 
$\pi_1(V) \ne 0$ then $V$ admits a finite cover $\tilde{V}$. (If $\pi_1(V)$ were finite, one could simply take the universal cover.  Residual finiteness becomes important when $\pi_1(V)$ is infinite.)

\medskip

The following lemma about submanifolds representing codimension one homology classes is well established for closed orientable manifolds; see \cite{Thom}.  When the ambient manifold has boundary the result appears well-known to the experts,  though we were unable to find an explicit statement or proof in the literature.  Hence we give it here.  The proof is a modification of the proof given in \cite[VI 11.16]{Bredon} for closed manifolds.

\begin{lem}\label{lem:realize}
Let $M$ be a smooth, compact, connected, oriented $n$--manifold with boundary.  Then every homology class in $H_{n-1}(M;\bbZ)$ is realized by a closed, smooth, oriented $(n-1)$--submanifold.
\end{lem}

\begin{proof}
Recall that we have the isomorphisms
\[H_{n-1}(M;\bbZ) \xleftarrow{\PD} H^1(M,\bdry M) \xleftarrow{} [(M,\bdry M),(S^1,p)]\]
where $p$ is a point in $S^1$.  Thus for any class $z \in H_{n-1}(M;\bbZ)$, there exists a map of pairs $\phi \colon (M, \bdry M) \to (S^1, p)$ such that $z=\PD(\phi^*(\xi))$, the Poincar\'e-Lefschetz dual to $\phi^*(\xi)$ where $\xi \in  H^1(S^1,p; \bbZ)\cong H^1(S^1;\bbZ) \cong \bbZ$ is the generator such that $\xi([S^1])=1$.  Then let $q\in S^1$ be a point different from $p$ and homotope $\phi$ to be transverse to $q$.  This gives us the submanifold $\phi^{-1}(q)$ which we now show represents $z$. Since $\PD(\xi)=[q] \in H_0(S^1;\bbZ)$, $\xi$ is also the Thom class of the normal bundle of $q$ in $S^1$.  Observing that $\phi$ gives a bundle map from the normal bundle of $\phi^{-1}(q)$ in $M$ to the normal bundle of $q$ in $S^1$, $\phi^*(\xi)$ is the Thom class of the normal bundle of $\phi^{-1}(q)$.  Thus $\PD(\phi^*(\xi)) = [\phi^{-1}(q)]$, and hence $z=[\phi^{-1}(q)]$.
\end{proof}

For our application, we only need the case when $n=3$.   A much simpler, standard, low-tech argument handles this case and applies whether or not  the manifold has boundary:    In a compact, orientable $3$--manifold $M$, an integral homology class $z \in H_2(M;\bbZ)$, viewed in terms of simplicial homology, can be represented as a integral linear combination of triangles ($2$--simplicies) in a triangulation of $M$, say $z=[\sum n_i \sigma_i]$ for integers $n_i$. In a tubular neighborhood of each $2$--simplex $\sigma$, regard $n\sigma$  as $|n|$ parallel copies of a triangle with disjoint interiors, oriented according to the sign of $n$, and glued together along the edges and vertices.   Since the linear combination is a cycle, in a tubular neighborhood along the edges but outside ball neighborhoods of the vertices we can match up the triangles according to their orientations and resolve the intersections. Now the resulting complex intersects the sphere boundaries of the ball neighborhoods of the vertices in collections of disjoint circles.  Replace the complex in these balls with disjoint disks bounded by these circles.  Since these resolutions and replacements preserve homology classes, the resulting surface represents the original homology class~$z$. 

\medskip
We now consider some basic definitions and facts about marginally outer trapped surfaces in initial data sets; for background, see e.g.\ \cite{AEM}.  Let $(V, h, K)$ be a $3$--dimensional initial data set in a $4$--dimensional 
spacetime $(M,g)$, i.e., $V$ is a  spacelike hypersurface in $M$ with induced (Riemannian) metric $h$ and second fundamental form $K$. To set sign conventions, for vectors $X,Y \in T_pV$, $K$ is defined as, 
$K(X,Y) = g(\D_X u, Y)$, where $\D$ is the Levi-Civita connection of $M$ and $u$ is the future directed timelike unit normal vector field to $V$. 

Let $\S$ be a closed (compact without boundary) two-sided hypersurface in $V$.   Then $\S$ admits a smooth unit normal field
$\nu$ in $V$, unique up to sign.  By convention, refer to such a choice as outward pointing. 
Then $l_+ = u+\nu$ (resp. $l_- =  u - \nu$) is a future directed outward (resp., future directed inward) pointing null normal vector field along $\S$. 

Associated to $l_+$ and $l_-$,  are the two {\it null second fundamental forms},  $\chi_+$ and 
$\chi_-$, respectively, defined as, 
\beq
\chi_{\pm} : T_p\S \times T_p\S \to \mathbb R , \qquad \chi_{\pm}(X,Y) =  g(\D_Xl_{\pm}, Y) \,.
\eeq
The {\it null expansion scalars} (or {\it null mean curvatures})  $\th_{\pm}$ of $\S$   are obtained by tracing 
$\chi_{\pm}$ with respect to the induced metric $\g$ on $\S$,
$\theta_{\pm} = {\rm tr}_{\g} \chi_{\pm}  = {\rm div}\,_{\S} l_{\pm}$.
Physically, $\th_+$ (resp., $\th_-$) measures the divergence
of the  outgoing (resp., ingoing) light rays emanating
from $\S$. 

The null expansion scalars can be expressed solely in terms of  the initial data 
$(V,g,K)$.  We have $\th_{\pm} = {\rm tr}_{\g} K \pm H$, 
where $H$ is the mean curvature of $\S$ within $M$.   In particular, in the time-symmetric case, i.e., when $K = 0$ (and hence $V$ is totally geodesic in $(M, g)$), $\theta_+$ is just the mean curvature of $\S$ in $V$.

In regions of spacetime where the
gravitational field is strong, one may have both $\th_- < 0$
and $\th_+ < 0$, in which case $\S$ is called a {\it trapped
surface}.   The concept of a trapped surface  plays a key role in the Penrose singularity theorem.   

Focusing attention on the outward null normal only, we say that
$\S$ is an outer trapped surface  if $\th_+ < 0$.  Finally, we define $\S$ to be a marginally
outer trapped surface (or MOTS) if $\th_+$ vanishes identically.   Note that in the time-symmetric case, a MOTS is simply a minimal hypersurface in $V$.   In this sense  MOTSs may be viewed as spacetime analogues of minimal surfaces.  Physically, MOTSs represent an extreme gravitational situation.  Under appropriate energy and causality conditions,  their presence in an initial data set signals the existence of a black hole.  
Moreover, for stationary (steady state) black hole spacetimes, cross sections of the event horizon are MOTSs.

Now consider a $3$--dimensional initial data set $(V,h,K)$, where $V$ is compact with boundary.  We say that a component $\S$ of 
$\d V$ is {\it null mean convex} if it has positive outward null expansion, $\th^+ > 0$, and negative inward null expansion, $\th^-< 0$.  Note that round spheres in Euclidean slices of Minkowski space  are null mean convex.

The basic existence results for MOTSs (see \cite{AEM} and references therein) imply the following result, which will be needed in the proofs of our main results.

\begin{lem}\label{exist}
Let $(V,h,K)$ be a $3$--dimensional initial data set, where $V$ is a compact $3$--manifold with boundary.  Suppose that the boundary can be expressed as a disjoint union, $\d V = \S_0 \cup \S_1$, such that each component of $\S_0$ is a MOTS  (with respect to either the null normal whose projection points into $V$ or the null normal whose projection points out of $V$)
and each component of $\S_1$ is null mean convex.  If there are no MOTS in $V \setminus \S_0$ then 
$\S_1$ is connected.
\end{lem}

%

Lemma \ref{exist} is a consequence of the basic existence result for MOTSs obtained in \cite{AM2}; cf.\,\cite[Theorem 3.1]{AM2}.  Moreover, as will be needed here, this existence result extends to the case of weak (nonstrict) barriers as described in \cite[Section 5]{AM2}.   We note also that Lemma \ref{exist} holds for initial data sets up to dimension seven, as a consequence of the higher dimensional existence results for MOTSs obtained in \cite{Eich1, Eich2}.

\proof[Proof of Lemma \ref{exist}]  
The argument is essentially contained in the proof of \cite[Theorem 5.1]{EGP}.  Suppose $\S_1$ has more than one component. Then we may express 
$\partial V$ as a disjoint union, $\d V = \S_{in} \cup \S_{out}$,  such that $\th_+ \leq 0$ along the components of $\S_{in}$, with respect to the inward pointing null normal, and with strict inequality on some component, and such that $\theta_+ \geq 0$ along the components of $\S_{out}$, with respect to the outward pointing null normal, and with strict inequality on some component.  (Take,  for example,  $\S_{in}$ to consist of one component of 
$\S_1$ and all components of $\S_0$ which are MOTS with respect to the inward pointing null normal, and take 
$\S_{out}$ to consist of the remaining components of $\S_1$ and all components of $\S_0$ which are MOTS with respect to the outward pointing null normal.)
Under these barrier conditions \cite[Theorem 3.1 and Section 5]{AM2} implies the existence of a MOTS $\S$ in $V$ homologous to $\S_0$.  Because of the strict barrier component, at least one component of $\S$ must be disjoint from $\S_0$ and is thus contained in $V \setminus \S_0$, contrary to our assumptions.  Hence, $\S_1$ must be connected.\qed

\medskip
A slightly more general notion of  MOTS was introduced in \cite{EGP}, and will be needed for our main results.

\begin{Def}
Given an initial data set $(V, g, K)$, we say that a subset $\S \subset V$ is an  \emph{immersed} MOTS if there exists a finite cover $p: \tilde V\rightarrow V$  and a  MOTS $\tilde \S$ in $\tilde V$ with respect to the pulled-back data $(p^{*}h, p^{*}K)$, such that $p(\tilde \S) = \S$.
\end{Def}

A simple example of an immersed MOTS (that is not a MOTS) occurs in the so-called 
$\bbR\bbP^3$ geon \cite{FSW, EGP}. 
The $\bbR\bbP^3$ geon is a globally hyperbolic spacetime that is double covered by the extended Schwarzschild spacetime.  Its Cauchy surfaces have the topology of  
$\bbR\bbP^3$ minus a point. 
The Cauchy surface $V$ covered by the totally geodesic time slice $\tilde V$ in the extended Schwarzschild spacetime has one asymptotically flat end (identical to an end in the Schwarzschild slice), and contains a projective plane $\S$ that is covered by the unique minimal sphere $\tilde \S$ in $\tilde V$.  Since $\S$ is not two-sided, it is not a MOTS. However, since the slice $\tilde V$ is totally geodesic, $\tilde \S$ is a MOTS, and hence
$\S$ is an immersed MOTS.

\section{The Main Result}
\label{results}

Consider a $3$--dimensional initial data set $(V, h,K)$, where $V$ is a compact, connected $3$--manifold with boundary $\d V = \S_0 \cup \S_1$, such that  $\S_0$ and $\S_1$ are orientable surfaces 
with no sphere components.\footnote{We assume throughout that $V$ is compact and connected, and that  
$\S_0$ and $\S_1$ are non-empty.}
In the  present context we are to think of $V$ as a compact spacelike hypersurface in the DOC of an ALADS black hole spacetime,  with $\S_0$ corresponding to a cross section of the event horizon and $\S_1$  corresponding to a surface ``near infinity".    At the initial data level, we represent $\S_0$  
by a marginally outer trapped surface (MOTS), and $\S_1$ by a null mean convex hypersurface. The null mean convexity condition follows from physically natural asymptotics in the ALADS setting; cf. \cite{CGS}. Moreover, since under suitable circumstances, there can be no (immersed) 
MOTSs\,\footnote{Such a surface would be visible from timelike infinity, but there are arguments precluding that possibility (see e.g.\ \cite{Wald, CGS}).} in the DOC, we will adopt this as an assumption on  $V \setminus \S_0$.  
 
\smallskip
The following is our main result.  For a  closed orientable surface $\S$, we let $g(\S)$ denote the sum of the genera of the components of $\S$.

\begin{thm}\label{main}
Let $(V, h,K)$ be a $3$--dimensional initial data set satisfying the following. 
\ben

\vs{-.5em}
\item[(a)] $V$ is a compact  manifold with boundary $\d V = \S_0 \cup \S_1$ where $\S_0$ and 
$\S_1$ are orientable surfaces, and each may have multiple  components, but no sphere components. 

\vs{-.5em}
\item[(b)] Each component of $\S_0$ is a MOTS (either with respect to  the inward pointing or outward pointing null normal\,\footnote{Thus we are allowing for both black holes and white holes.})
and $\S_1$ is null mean convex.  

\vs{-.5em}
\item[(c)] There are no immersed MOTS in  $V \setminus \S_0$.
\een
Then $\Sigma_1$ is connected and one of the following holds. 

\smallskip
\ben
\vs{-.5em}
\item[(i)] $\Sigma_0$ is connected, $g(\Sigma_0) = g(\Sigma_1)$, and $V$ is diffeomorphic to  $\Sigma_0 \times [0,1]$.
\vs{-.5em}
\item[(ii)] $\Sigma_0$ has multiple components, $g(\Sigma_0) = g(\Sigma_1)$, and $V$ is diffeomorphic to  $\Sigma_0 \times [0,1]$ with $1$--handles attached to $\Sigma_0 \times\{1\}$.
\vs{-.5em}
\item[(iii)] $\Sigma_0$ may have multiple components, $g(\Sigma_0) < g(\Sigma_1)$, and $V$ is diffeomorphic to $ \Sigma_0 \times [0,1]$ with $1$--handles attached to $\Sigma_0 \times\{1\}$.
\een
\end{thm}

\begin{figure}
\centering
\includegraphics[width=\textwidth]{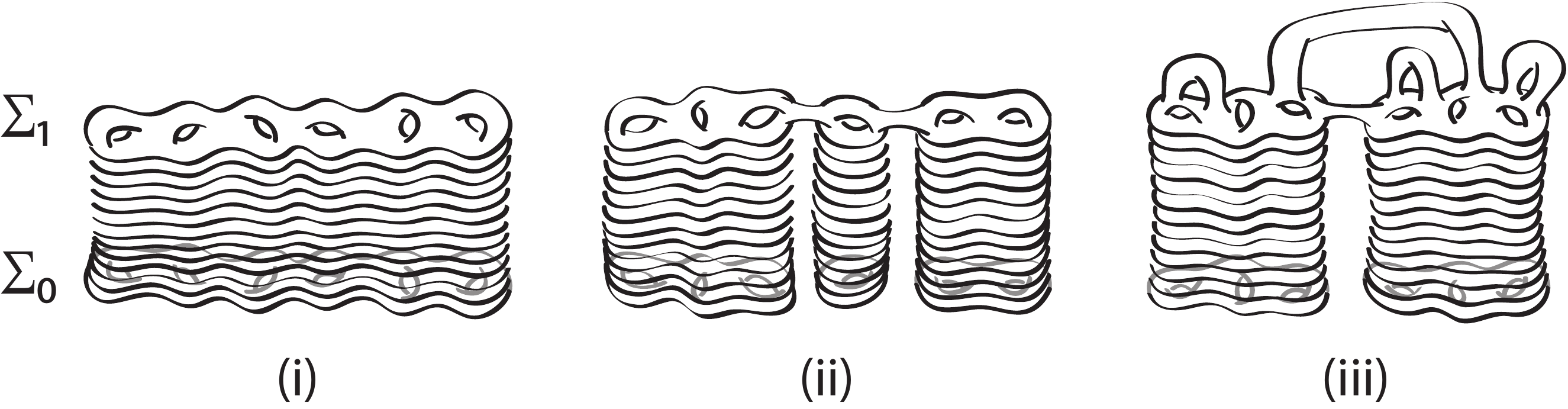}
\caption{Examples of the three possible topology types from Theorem~\ref{main}.}
\label{fig:mainthm}
\end{figure}

Note that in all three cases, $g(\S_0) \le g(\S_1)$.  This is consistent with the genus inequality obtained in \cite[Theorem 4.1]{GSWW} in the spacetime setting, as well as a more restrictive inequality obtained in \cite{lee-neves}.  Theorem \ref{main} gives a complete description of the possible topologies when equality   holds and when the inequality is strict. Figure~\ref{fig:mainthm}  illustrates these possibilities; these are examples of what are called {\em compression bodies}.  Case~(i) is realized by suitably truncated time slices in the generalized Kottler spacetimes.  It is an open question whether the more complicated topologies of cases (ii) and (iii) can be physically realized.  
 In the static case, a uniqueness result for negative mass Kottler spacetimes is obtained in \cite{lee-neves}, assuming, in the notation used here, that 
the genus of a component of $\S_0$  equals the genus of $\S_1$.  If $\S_0$ is connected, the following immediate consequence of Theorem \ref{main} guarantees {\it topological} uniqueness in this case.


\begin{cor}\label{cor}
Let $(V, h,K)$ be a $3$--dimensional initial data set such that $V$ is a compact  manifold with boundary $\d V = \S_0 \cup \S_1$, where $\S_0$ and 
$\S_1$ are orientable surfaces, with $\S_0$ connected 
and with $g(\S_0) \ge g(\S_1) \geq 1$. Suppose that  $\S_0$ is a MOTS and $\S_1$ is null mean convex.  If there are no immersed MOTS in  $V \setminus \S_0$  then, in fact,  $g(\S_0) = g(\S_1)$ and $V$ is diffeomorphic to  $\Sigma_0 \times [0,1]$.
\end{cor}

\noindent
{\it Remark:} If $\S_0$ is not connected (but has no sphere components), the conclusion of this corollary is still true provided the genus assumption $g(\S_0) \ge g(\S_1)$ is replaced by: $g(\S') \ge g(\S_1)$, where $\S'$ is a component of $\S_0$.  Indeed, this assumption clearly rules our cases (ii) and (iii) of Theorem \ref{main}.  If one further allows $\S_0$ to have some sphere components, then, by a modification of our arguments, one could conclude that 
$V$ is topologically a product minus a finite number of balls.

\smallskip
Theorem \ref{main} follows from Lemmas \ref{props} and \ref{topo} presented below.  Lemma \ref{props} establishes properties of $V$ that make explicit use of the geometric assumptions, e.g, that there are no immersed MOTSs in $V\setminus \S_0$.  Lemma \ref{topo} is a purely topological result derived from a key result in \cite{CLR}.

\begin{lem}\label{props}  
Let $(V, h,K)$ be a $3$--dimensional initial data set satisfying the following: 
\ben

\item[(a)] $V$ is a compact  manifold with boundary $\d V = \S_0 \cup \S_1$ where $\S_0$ and $\S_1$ are orientable surfaces, and each may have multiple components, but no sphere components. 

\item[(b)] Each component of $\S_0$ is a MOTS (either with respect to  the inward pointing or outward pointing null normal) and $\S_1$ is null mean convex.  

\vs{-.5em}
\item[(c)] There are no immersed MOTS in  $V \setminus \S_0$.
\een
Then,

\ben

\vs{-.3em}
\item[(i)] $V$ is orientable.

\item[(ii)] There are no non-separating closed surfaces in any finite cover of $V$.

\vs{-.3em}
\item[(iii)] $V$ is irreducible.

\vs{-.3em}
\item[(iv)] $\Sigma_0$ is incompressible. 

\vs{-.3em}
\item[(v)] $\Sigma_1$ is  connected. 
\een
\end{lem}

\begin{proof}
We show that if any one of (i), (ii), (iii), (iv) is not satisfied, then there exists a finite cover $p\colon \tilde V \to V$ such that $p^{-1}(\S_1)$ is not connected.  Consequentially, Lemma \ref{exist} then implies that there is an immersed MOTS in $V \setminus \S_0$, contrary to assumption. We immediately obtain (v) by taking the trivial cover in Lemma \ref{exist}.

(i) If $V$ is non-orientable, then $V$ has an orientable double cover.  For such covers, each boundary component of $\bdry V$, since assumed orientable, lifts to two components.  In particular, the lift of $\Sigma_1$ is not connected.

(ii) Assume $S$ is a non-separating closed surface in $V$. Without loss of generality we may assume $S$ is connected.  Form a double cover $p \colon \tilde{V} \to V$ as follows:  Let $W$ be the result of cutting $V$ open along $S$.  Observe that $\bdry W = \bdry V \cup \tilde{S}$ where $\tilde{S}$ is the double cover of $S$ arising as the boundary of a collar neighborhood of $S$ in $V$. Let $\phi \colon \tilde{S} \to \tilde{S}$ be the corresponding deck transformation.  Then let $(W_i,\tilde{S}_i)$ be a copy of $(W,\tilde{S})$ for each $i=1,2$ and form $\tilde{V} = W_1 \cup W_2 / \tilde{\phi}$ where $\tilde{\phi} \colon \tilde{S}_1 \to \tilde{S}_2$
 is the orientation reversing homeomorphism induced by $\phi$.  By construction $p^{-1}(\Sigma_1)$ is not connected.

If instead $S$ were a non-separating closed surface in a finite cover $V'$ of $V$, then as above there would be a double cover $\tilde{V}'$ of $V'$.   Composing these covers gives a finite cover $p \colon \tilde{V}' \to V$ for which $p^{-1}(\Sigma_1)$ is not connected.

(iii)  Assume $V$ is reducible.  Since $\bdry V \neq \emptyset$, this implies $V = V' \# V''$ is a non-trivial connected sum.  (The only prime but reducible $3$--manifolds are $S^1 \times S^2$ and $S^1 \tilde{\times} S^2$ \cite[Lemma 3.13]{Hempel}.)  Since $\Sigma_1$ is connected, it is a boundary component of, say, $V'$.   Since, as discussed in Section \ref{prelim},  every compact $3$--manifold is residually finite, $V''$ has a cover of finite index $k>1$.  This induces a finite cover $p \colon \tilde{V} \to V$ in which the $V'$ summand lifts to $k$ disjoint copies.  Hence $p^{-1}(\Sigma_1)$ is not connected.

(iv) Assume $\Sigma_0$ is compressible.  Thus there is a compressing disk $D$ for $\Sigma_0$, and $D$ is necessarily disjoint from $\Sigma_1$.  If $D$ is non-separating in $V$, then we may form a double cover as in (ii) with $D$ playing the role of $S$.  On the other hand, if $D$ separates $V$ into $V'$ and $V''$ where $\Sigma_1$ is a boundary component of $V'$ (recall that $\Sigma_1$ is connected), then we may form a finite cover of $V$ extending a finite cover of $V''$ as in (iii).  In either case, $\Sigma_1$ lifts to a disconnected surface.
\end{proof}

\begin{lem}\label{topo}
Let $V$ be a compact, connected, orientable, irreducible $3$--manifold with $\bdry V=\Sigma_0\cup\Sigma_1$ such that $\Sigma_1$ is connected, $\Sigma_0$ is incompressible and potentially disconnected, 
and $\bdry V$ has no sphere components.  Assume there are no non-separating closed surfaces in any finite cover of $V$.

Then, using $|\Sigma_0|$ to denote the number of components of $\Sigma_0$,
\begin{itemize}
\item  $g(\Sigma_0) \leq g(\Sigma_1)$ and
\item  $V \cong \Sigma_0 \times [0,1]$ with at least $|\Sigma_0|-1$ $1$--handles attached to $\Sigma_0 \times\{1\}$.
\end{itemize}
Moreover $g(\Sigma_0) = g(\Sigma_1)$ if and only if exactly $|\Sigma_0|-1$ $1$--handles are attached to $\Sigma_0 \times\{1\}$. 
In particular, if $\Sigma_0$ is connected, then $V \cong \Sigma_0 \times [0,1]$.
\end{lem}

\begin{proof}
The proof splits according to whether or not $\Sigma_1$ is incompressible.

\smallskip
\noindent {\bf Case 1.}  $\Sigma_1$ is incompressible.

Corollary 3.6 in \cite{CLR} then shows that either (a) $V$ is covered by a product $F \times [0,1]$ for some surface $F$ or (b) there exists a finite cover $\tilde{V}$ of $V$ such that $\rk(H_2(\tilde{V})/\incl_*(H_2(\bdry \tilde{V})))>0$, where $\incl_* \colon  H_2(\bdry \tilde{V})) \to H_2(\tilde{V})$ is the homomorphism induced by inclusion.  

%
\smallskip
\noindent(a)  There is a cover $p \colon \tilde{V} \cong F \times [0,1] \to V$ for some surface $F$.
  Let $\tilde{\Sigma}_i$ be the lift of $\Sigma_i$ to $\tilde{V}$ for each $i=0,1$.  Since $\bdry \tilde{V}$ has just two boundary components, $\tilde{\Sigma}_i$ is connected for each $i=0,1$.  Hence both $\Sigma_0$ and $\Sigma_1$ are connected.  

Fix a basepoint $* \in \Sigma_0$.  By incompressibility, the Loop Theorem implies the inclusion map $\Sigma_0 \hookrightarrow V$ induces an injection $\pi_1(\Sigma_0, *)$  into $\pi_1(V, *)$.  We claim that this is also a surjection:  A loop $\gamma \subset V$ based at $*$ lifts under $p$ to  a loop $\tilde{\gamma}$ connecting two points of $p^{-1}(*) \subset \tilde{\Sigma_i}$.  Since $\tilde{V} \cong F \times [0,1]$, $\tilde{\gamma}$ is homotopic to  a path $\tilde{\gamma}'$ connecting the same points of $p^{-1}(*)$.  Thus $p(\tilde{\gamma})$ is homotopic in $V$, fixing endpoints in $* \in \Sigma_0$, to $\gamma$.  Thus $\pi_1(\Sigma_0) \cong \pi_1(V)$. 
Hence, we are in the first case of   Theorem~\ref{thm:product}, which then  implies that $V \cong \Sigma_0 \times [0,1]$.  Consequentially,  $g(\Sigma_0) = g(\Sigma_1)$.

\smallskip
\noindent(b) There is a finite cover $p\colon \tilde{V} \to V$  such that $\rk(H_2(\tilde{V})/\incl_*(H_2(\bdry \tilde{V})))>0$.   Hence there is a non-trivial element in $H_2(\tilde{V})$ that is not in $\incl_*(H_2(\bdry \tilde{V}))$.
By Lemma~\ref{lem:realize}, such an element can be realized by a closed surface $S$ embedded in $\tilde{V}$.  Since any connected, closed, separating surface in $\tilde{V}$ represents a class in $\incl_*(H_2(\bdry \tilde{V}))$, $S$ must have a non-separating component.  This is contrary to our assumptions.

\medskip
\noindent {\bf Case 2.}  $\Sigma_1$ is compressible. 

  
It is always possible to choose a minimal set of mutually disjoint compressing disks for $\Sigma_1$ such that the compressions of $\Sigma_1$ produces a manifold $V' \subset V$ with incompressible boundary $\Sigma_0 \cup \Sigma_1'$. Let $D_1, \dots, D_n$ be such a set.  
Then $V$ may be recovered from $V'$ by attaching disjoint $1$--handles $h_1, \dots, h_n$ along $\Sigma_1'$; the compressing disk $D_i$ may be identified with the co-core of the $1$--handle $h_i$.  For each $1$--handle $h_i$, let $d_i$ and $e_i$ be the two disks in $\Sigma_1'$ to which it is attached.

We claim that $\Sigma_1'$ has no sphere components:  Otherwise, since $V$ is $V'$ with $1$--handles attached, $\bdry V$ has no sphere components, and $V$ is irreducible, a sphere component  of $\Sigma_1'$ would have to bound a $3$--ball $B$ in $V'$ ($V \setminus V'$ is just the $1$--handles).  Because $\Sigma_0 \neq \emptyset$ and $\Sigma_1$ is connected, there must be a $1$--handle, say $h$, connecting $\bdry B$ to another component of $\Sigma_1'$.  But then $V' \cup h \cup B \cong V'$, implying that the compression corresponding to $h$ was not necessary.  This is contrary to the presumed minimality of our set of compressing disks.  

Then since $V$ is irreducible, so also must be $V'$.  However, it may be that $V'$ is disconnected even though $V$ is connected. 
Let the components of $V'$ be $W^1, \dots, W^r$.   Then for each $s =1, \dots, r$, let $W_0^s = W^s \cap \Sigma_0$ and $W_1^s = W^s \cap \Sigma_1'$ so that $\bdry W^s = W_0^s \cup W_1^s$.  

We claim that each component is a product, $W^s \cong W_0^s \times [0,1]$.  If this is so, then $V' \cong \Sigma_0 \times [0,1]$ and thus $V$ is the product $\Sigma_0 \times [0,1]$ with $1$--handles attached to $\Sigma_0 \times \{1\}$.  Because $\Sigma_1$ is connected, a total of $|\Sigma_0|-1$ $1$--handles are required to connect all the components of $V'$.  These $1$--handles do not alter the total genus of the boundary.  Any further attachment of $1$--handles increases the genus.  The conclusion of the theorem then follows. Therefore the remainder of this proof shows that $W^s \cong W_0^s \times [0,1]$ for each $s=1, \dots, r$.

\smallskip
 Because $W^s$ is irreducible and has incompressible boundary with no sphere components, we may appeal to \cite[Corollary 3.6]{CLR} again and follow a similar argument as in Case 1 for each $s =1, \dots, r$.

\smallskip
\noindent(a) There is a cover $p^s \colon \tilde{W^s} \cong F \times [0,1] \to W^s$.  As before, elementary covering space arguments and Theorem~\ref{thm:product} imply that $W^s$ is an $I$--bundle over a surface.  Thus $\bdry W^s$ has at most two components. By construction, we necessarily have $W_1^s \neq \emptyset$, but it is possible that $W_0^s= \emptyset$.    If $W_0^s \neq \emptyset$, then both $W_0^s$ and $W_1^s$ are connected and we again conclude as in Case 1 that $W^s \cong W_0^s \times [0,1]$ and $g(W_0^s)=g(W_1^s)$.  This is the desired conclusion.

If  $W_0^s = \emptyset$ and $W_1^s$ has two components, then we again conclude that the two components have the same genus and $W^s$ is the product of a closed surface and an interval.  In particular, there is a closed surface $S$ in the interior of $W^s$
that separates the components of $W_1^s$.   Since $\Sigma_1$ is connected, there is a path joining the two components of $W_1^s$ that runs through a subset of the $1$--handles $h_1, \dots, h_n$ and the components of $\Sigma_1' \setminus  W_1^s$.  Hence in $V$ the surface $S$ is non-separating, contrary to assumption.

 If  $W_0^s = \emptyset$ and $W_1^s$ has just one component, then $W^s$ is a twisted $I$--bundle over a closed, compact, non-orientable surface $S$.  Since $W^s$ is orientable, $S$ is one-sided and hence non-separating in $W^s$.  Therefore it remains non-separating in $V$, a contradiction.  
 
\smallskip
\noindent(b) There is a finite cover $p^s \colon \tilde{W^s} \to W^s$ of degree $k$ such that $\rk(H_2(\tilde{W^s})/\incl_*(H_2(\bdry \tilde{W^s})))$ $>0$.
Again there is a non-separating, closed surface $S$ embedded in $\tilde{W^s}$.  

We now extend this to a finite covering of $V$ in which $S$ continues to be non-separating.  First extend $p^s$ to a degree $k$ cover $p' \colon \tilde{V'} \to V'$ where $\tilde{V'}$ is the disjoint union of $\tilde{W^s}$ with $k$ copies of $V' \setminus W^s$ and $p' \vert (\tilde{V'} \setminus \tilde{W^s})$ is the trivial $k$--fold covering map.
Then each of the disks $d_i$ and $e_i$ in $\Sigma_1'$ (to which the handle $h_i$ is attached) has $k$ preimages in $\tilde{\Sigma_1'} \subset \bdry \tilde{V'}$. Denote the preimage of $d_i$ as $\tilde{d_i}^1 \cup \dots \cup \tilde{d_i}^k$ and the preimage of $e_i$ as $\tilde{e_i}^1 \cup \dots \cup \tilde{e_i}^k$.  For each $i=1, \dots, n$,  let $\tilde{h_i}^1\cup \dots \cup \tilde{h_i}^k$ be $k$ disjoint $1$--handles covering $h_i$; then for each $j=1,\dots,k$, attach the $1$--handle $\tilde{h_i}^j$ to $\tilde{V'}$ along the disks $\tilde{d_i}^j$ and $\tilde{e_i}^j$ in $\tilde{\Sigma_1'}$. By construction, this gives a covering map $p \colon \tilde{V} \to V$ that restricts to our initial covering map $p' \colon \tilde{V'} \to V'$.  In particular, $S$ continues to be non-separating in $\tilde{V}$, contrary to our assumptions.
\end{proof}
Lemmas \ref{props} and \ref{topo} combine to give a proof of Theorem \ref{main}.
\proof[Proof of Theorem \ref{main}] Lemma \ref{props} implies that all the assumptions of Lemma \ref{topo} hold.  The conclusion of Lemma \ref{topo} implies that of Theorem \ref{main}.\qed

\medskip
\noindent\emph{Acknowledgements.} The authors would like to thank Piotr Chru\'sciel for valuable comments related to this work and Nikolai Saveliev for his input. 
The work of KLB was supported by Simons Foundation Collaboration Grant \#209184. 
 The work of GJG was supported by NSF grant DMS-1313724 and by a Fellows Program grant from the Simons Foundation (Grant No.  63943). 


\providecommand{\bysame}{\leavevmode\hbox to3em{\hrulefill}\thinspace}
\providecommand{\MR}{\relax\ifhmode\unskip\space\fi MR }
\providecommand{\MRhref}[2]{%
  \href{http://www.ams.org/mathscinet-getitem?mr=#1}{#2}
}
\providecommand{\href}[2]{#2}

\end{document}